\documentclass[11pt]{amsart}
\usepackage{amsmath,amsthm,amsfonts,amssymb}
\usepackage{verbatim}
\usepackage{latexsym}
\usepackage{nicefrac}
\usepackage{color}
\usepackage{ulem}
\usepackage{enumerate}
\usepackage{mdwlist}
\usepackage[british]{babel}
\usepackage[utf8]{inputenc}
\usepackage{soul}
\usepackage{bbm}

\newcommand{\R}{{\mathbb{R}}}
\newcommand{\C}{{\mathbb{C}}}

\newcommand{\E}{\mathbb{E}}
\newcommand{\Prob}{\mathbb{P}}

\DeclareMathOperator{\tr}{tr}
\DeclareMathOperator{\sym}{Sym}

\newtheorem{theorem}{Theorem}[subsection]
\newtheorem{lemma}[theorem]{Lemma}

\newtheorem{definition}[theorem]{Definition}

\newtheorem{remark}[theorem]{Remark}

\begin{document}
	
	\title{Stable low-rank matrix recovery from 3-designs}
	\author{Timm Gilles}
	\email{gilles@mathga.rwth-aachen.de}
	\address{Lehrstuhl f\"ur Geometrie und Analysis, RWTH Aachen University, D-52056 Aachen, Germany}
	
	\maketitle
	\begin{abstract}
		We study the recovery of low-rank Hermitian matrices from rank-one measurements obtained by uniform sampling from complex projective 3-designs, using nuclear-norm minimization. This framework includes phase retrieval as a special case via the PhaseLift method. In general, complex projective $t$-designs provide a practical means of partially derandomizing Gaussian measurement models. While near-optimal recovery guarantees are known for $4$-designs, and it is known that $2$-designs do not permit recovery with a subquadratic number of measurements, the case of $3$-designs has remained open. In this work, we close this gap by establishing recovery guarantees for (exact and approximate) $3$-designs that parallel the best-known results for $4$-designs. In particular, we derive bounds on the number of measurements sufficient for stable and robust low-rank recovery via nuclear-norm minimization. Our results are especially relevant in practice, as explicit constructions of $4$-designs are significantly more challenging than those of $3$-designs.	
	\end{abstract}

	\section{Introduction}
	\setcounter{subsection}{1} 
	We consider the following well-studied instance of the low-rank matrix recovery problem: Let $X\in\C^{n\times n}$ be an unknown Hermitian low-rank matrix. One observes the measurements 
	\begin{align*}
	\tr(XA_j),\quad j=1,\ldots,m,
	\end{align*}
	where $A_j=a_ja_j^{\ast}$ for some vectors $a_j\in\C^n$. Then, the goal is to recover the matrix $X$ from as few measurements as possible. We can also express this process in terms of a so-called measurement operator, defined as
	\begin{align}\label{eq:def_measurement_operator}
		\mathcal{A}\colon \mathcal{H}_n\to \R^m,Z\mapsto \sum_{j=1}^m \tr(Z A_j )e_j,
	\end{align}
	where $\mathcal{H}_n$ denotes the set of $n\times n$ Hermitian matrices. Then, the (possibly noisy) measurement process is given by 
	\begin{align}\label{eq:def_meas_process}
		b=\mathcal{A}(X)+\epsilon,
	\end{align}
	where $\epsilon\in\R^m$ denotes additive noise satisfying $\Vert \epsilon \Vert_{\ell_q}\leq \eta$. It is natural to ask whether there exists a recovery scheme that efficiently recovers $X$ and whether fewer than $n^2$ measurements are sufficient. One well-studied approach is the nuclear norm minimization problem, a convex minimization problem which is given by
	\begin{align}\label{eq:def_nuclear_norm minimization}
		\min_{Z\in \mathcal{H}_n}\Vert Z\Vert_{\ast}\quad \text{subject to}\quad \Vert\mathcal{A}(Z)-b\Vert_{\ell_q}\leq \eta,
	\end{align}
	where $\Vert\cdot\Vert_{\ast}$ denotes the nuclear norm. There are multiple results considering different random measurement processes of the above type that show that the minimization problem (\ref{eq:def_nuclear_norm minimization}) efficiently recovers $X$ using fewer than $n^2$ measurements with high probability. This holds e.g. if the vectors $a_1,\ldots,a_m$ are independent complex Gaussian random vectors \cite{candes2014solving,tropp2015convex,kueng2017low,kabanava2016stable} or are independently and uniformly drawn from a (approximate) $t$-design \cite{gross2015partial,kueng2017low,kabanava2016stable}.
	
	\begin{definition}
		Let $\{w_1,\ldots,w_N\}\subseteq \C^n$ be a set of normalized vectors with corresponding weights $\{p_1,\ldots,p_N\}\subseteq [0,1]$ such that $\sum_{i=1}^N p_i=1$. The set $\{p_i,w_i\}_{1\leq i\leq N}$ is called an approximate $t$-design of $p$-norm accuracy $\theta_p$, if
		\begin{align*}
			\left\Vert \sum_{i=1}^N p_i(w_iw_i^{\ast})^{\otimes t}-\binom{n+t-1}{t}^{-1} P_{\sym^t} \right\Vert_p\leq \binom{n+t-1}{t}^{-1} \theta_p,
		\end{align*}
		where $P_{\sym^t}\colon (\C^n)^{\otimes t}\to (\C^n)^{\otimes t}$ projects onto the totally symmetric subspace $\sym^t$ of $(\C^n)^{\otimes t}$, is defined by
		\[
		P_{\sym^t}(z_1\otimes \ldots \otimes z_t)=\frac{1}{t!}\sum_{\pi\in S_t} z_{\pi(1)}\otimes \ldots \otimes z_{\pi(t)}
		\]
		and is called symmetrizer map. If $\theta_p=0$, then $\{p_i,w_i\}_{1\leq i\leq N}$ is called a weighted complex projective $t$-design.
	\end{definition}
	
	Let $\{p_i,w_i\}_{1\leq i\leq N}$ be a $t$-design. Then, it holds that
	\begin{align*}
	\mathbb{E}_{a\sim \{p_i,w_i\}_{1\leq i\leq N}}[(aa^{\ast})^{\otimes t}]=\binom{n+t-1}{t}^{-1} P_{\sym^t} =\int_{\Vert w\Vert_{\ell_2}=1} (ww^{\ast})^{\otimes t} \,\text{d}w=\mathbb{E}_{a\sim \mathcal{U}(S^{n-1})}[(aa^{\ast})^{\otimes t}],
	\end{align*}
	which follows from \cite[Lemma 1]{kueng2017low} and where the integral is with respect to normalized Haar measure on the sphere. Thus, the first $2t$-moments (understood in the above sense) of a random vector that uniformly samples from a $t$-design are equal to the moments when sampling uniformly from the sphere, which again is the same as sampling according to a standard Gaussian vector and normalizing each sample \cite{kueng2017low}. This explains why $t$-designs are often used to partially derandomize a process that is e.g., based on a complex Gaussian random vector; see \cite{gross2015partial}.\\
	The main objective of this paper is to establish a stable and robust recovery result via nuclear norm minimization for the measurement process described in (\ref{eq:def_meas_process}), where the $a_j$ are independently and uniformly drawn from a (approximate) $3$-design. By proving such a result, we close an open gap in the literature that was first noted in \cite{kueng2017low}. Addressing this gap is particularly relevant, since explicit constructions of $4$-designs are significantly more challenging and arguably offer less structure than those of $3$-designs; see \cite{ambainis2007quantum,zhu2017multiqubit,kueng2016low}.\\
	We finish this introduction by remarking that plenty of recovery results have also been proven for the measurement process (\ref{eq:def_meas_process}) in which $A_j$ is not assumed to be of rank one; see \cite{candes2012exact,candes2011tight,gross2011recovering,liu2011universal,tropp2015convex,kabanava2016stable}.
	
	\section{Main results}
	\setcounter{subsection}{1} 
	Our main theorems establish recovery guarantees under the assumption that a random vector $a_j$ is drawn from a (approximate) $3$-design $\{p_i,w_i\}_{1\leq i\leq N}$. By that we mean that the vector $w_i$ is picked with probability $p_i$. Before stating our main results, we introduce some notation. For any $Z\in\C^{n\times n}$, let $Z_r=\sum_{j=1}^r\sigma_j(Z)u_jv_j^{\ast}$ and $Z_c=\sum_{j=r+1}^n\sigma_j(Z)u_jv_j^{\ast}$, where $Z=\sum_{j=1}^n\sigma_j(Z)u_jv_j^{\ast}$ is the singular value decomposition of $Z$, and the singular values $\sigma_j(Z)$ are arranged in decreasing order.

	\begin{theorem}\label{th:general_recovery_from_exact_3_designs}
		Let $q\geq 1$. Let $\{p_i,w_i\}_{1\leq i\leq N}$ be a weighted 3-design and let $a_1,\ldots, a_m\in \C^m$ be independently drawn from $\{p_i,w_i\}_{1\leq i\leq N}$. Consider the measurement process described in (\ref{eq:def_measurement_operator}) with $A_j=\sqrt{n(n+1)}\,a_j a_j^{\ast}$ for $1\leq j\leq m$. Fix $1\leq r\leq n, 0< \rho <1$ and assume that
	 	\begin{align*}
	 		m\geq C_1\rho^{-6} r^3 n \log(2n).
	 	\end{align*}
 		Then, with probability at least $1-\exp\left(-\frac{C_2 m}{\rho^{-4}r^2}\right)$ it holds that for any $X\in\mathcal{H}_n$, the solution $X^{\#}$ to the convex optimization problem (\ref{eq:def_nuclear_norm minimization}) with $b=\mathcal{A}(X)+\epsilon$, where $\Vert \epsilon \Vert_{\ell_q}\leq \eta$, satisfies
 	\begin{align*}
 		\Vert X-X^{\#}\Vert_F\leq \frac{2(1+\rho)^2}{(1-\rho)\sqrt{r}}\Vert X_c\Vert_{\ast}+\frac{2C_3(3+\rho)r}{(1-\rho)\rho^2}\cdot\frac{\eta}{m^{1/q}}.
 	\end{align*}
 	Here, $C_1,C_2,C_3>0$ are absolute constants. 
	\end{theorem}
	
	We now elaborate on the above result and place it within the current research landscape. $t$-designs were first used as a tool to partially derandomize a Gaussian measurement process in \cite{gross2015partial}, where the authors studied the phase retrieval problem via the PhaseLift approach. This problem can be viewed as a special case of the measurement process in (\ref{eq:def_meas_process}) in which the unknown matrix $X$ is rank-one and positive semidefinite. In that setting, they proved a non-uniform recovery result showing that
		\begin{align}\label{eq:number_of_meas_gross_krahmer_paper}
			m\gtrsim t n^{1+\frac{2}{t}}\log(n)^2
		\end{align}
		randomly chosen measurements from a $t$-design with $t\geq 3$ enable recovery via a convex minimization problem similar to (\ref{eq:def_nuclear_norm minimization}); see \cite[Theorem 1]{gross2015partial}. This result was later improved and generalized to a stable and robust uniform low-rank matrix recovery guarantee requiring only $\mathcal{O}( r n \log(n))$ independently drawn measurements from a $t$-design with $t\geq 4$.\\
	Our Theorem \ref{th:general_recovery_from_exact_3_designs} provides an analogous result for $3$-designs. For phase retrieval via PhaseLift, it improves the required number of measurements by removing the suboptimal exponent $1+\frac{2}{3}$ of $n$ as well as one logarithmic factor in (\ref{eq:number_of_meas_gross_krahmer_paper}). For the general low-rank matrix recovery problem, we believe that the exponent $3$ of $r$ is an artifact of the proof and conjecture that it can be reduced to $1$ using alternative proof techniques.\\
	Note that our Theorem \ref{th:general_recovery_from_exact_3_designs} is closely related to the main result \cite[Theorem 1]{kueng2016low}. In that work, the authors established a uniform recovery guarantee requiring the same number of measurements, where the measurement vectors are drawn independently from Clifford orbits. Since these orbits form complex projective $3$-designs, their setting is a special case of ours; see \cite{zhu2017multiqubit,webb2015clifford}. 
	\\
	Recovery from approximate $3$-designs is also possible; see and compare with \cite{kabanava2016stable}.
	
	\begin{theorem}\label{th:general_recovery_from_approx_3_designs}
		Let $\{p_i,w_i\}_{1\leq i\leq N}$ be an approximate 3-design of accuracy either $\theta_{1}\leq\frac{1}{4}$ or $\theta_{\infty}\leq \frac{1}{4r(1+(1+\rho^{-1})^2)}$. Assume that
		\begin{align*}
			\left\Vert \sum_{i=1}^N p_i w_i w_i^{\ast}-\frac{1}{n}\text{id} \right\Vert_{\infty}\leq \frac{1}{n},
		\end{align*}
		
		and let $a_1,\ldots, a_m\in \C^m$ be independently drawn from $\{p_i,w_i\}_{1\leq i\leq N}$. Consider the measurement process described in (\ref{eq:def_measurement_operator}) with $A_j=\sqrt{n(n+1)}\,a_j a_j^{\ast}$ for $1\leq j\leq m$. Fix $1\leq r\leq n, 0< \rho <1$ and assume that
		\begin{align*}
			m\geq c_1 \rho^{-6} r^3 n \log(2n).
		\end{align*}
		Then, with probability at least $1-\exp\left(-\frac{c_2 m}{\rho^{-4}r^2}\right)$ it holds that for any $X\in\mathcal{H}_n$, the solution $X^{\#}$ to the convex optimization problem (\ref{eq:def_nuclear_norm minimization}) with $b=\mathcal{A}(X)+\epsilon$, where $\Vert \epsilon \Vert_{\ell_q}\leq \eta$, satisfies
		\begin{align*}
			\Vert X-X^{\#}\Vert_F\leq \frac{2(1+\rho)^2}{(1-\rho)\sqrt{r}}\Vert X_c\Vert_{\ast}+\frac{2c_3(3+\rho)r}{(1-\rho)\rho^2}\cdot\frac{\eta}{m^{1/q}}.
		\end{align*}
		Here, $c_1,c_2,c_3>0$ are absolute constants. 
	\end{theorem}
	
	\begin{remark}\label{rem:equivalent_formulation_for_super_normalized_design}
		One can state Theorem \ref{th:general_recovery_from_exact_3_designs} and Theorem \ref{th:general_recovery_from_approx_3_designs} equivalently for so called super-normalized weighted 3-design, i.e. $\{p_i,\widetilde{w}_i\}_{1\leq i\leq N}$ with
		\begin{align*}
			\widetilde{w}_i=\sqrt[4]{(n+1)n} \,w_i.
		\end{align*}
		Then, for $a_1,\ldots,a_m$ independently drawn from $\{p_i,\widetilde{w}_i\}_{1\leq i\leq N}$ one would define the measurement matrices $A_j$ as $A_j=a_j a_j^{\ast}$. This gives the same measurement process as in Theorem \ref{th:general_recovery_from_exact_3_designs} resp. Theorem \ref{th:general_recovery_from_approx_3_designs}; see also \cite{kueng2017low} and \cite{kabanava2016stable}.
	\end{remark}
	
	\section{Proofs}
	
	The proof structure for Theorem \ref{th:general_recovery_from_exact_3_designs} and Theorem \ref{th:general_recovery_from_approx_3_designs} follows the ideas presented in \cite{kabanava2016stable}, which themselves are based on Mendelson's small ball method \cite{koltchinskii2015bounding, mendelson2015learning}. Their proof uses the Payley-Zygmund inequality at one point, which creates the need to bound a fourth moment - something that is only possible when sampling from an (approximate) $4$-design. To avoid this, we use a variant of the Payley-Zygmung inequality that requires bounding only a third moment, which can be done when sampling from a $3$-design.\\
	While the core ideas align with existing arguments, our main contributions lie in the refined estimates presented in Lemma \ref{lem:Q_bound} and Lemma \ref{lem:gen_bound_for_Q}. For clarity and self-containment, we provide a full proof of Theorem \ref{th:general_recovery_from_exact_3_designs} and Theorem \ref{th:general_recovery_from_approx_3_designs}.
	
	\subsection{Proof of Theorem \ref{th:general_recovery_from_exact_3_designs}}
	
	Since some calculations of the proof of Theorem \ref{th:general_recovery_from_exact_3_designs} can be written more concisely for super-normalized designs, we assume in the following the setting described in Remark \ref{rem:equivalent_formulation_for_super_normalized_design}.
	As already stated this proof follows the same outline as the one presented in \cite[Theorem 1.3]{kabanava2016stable} which itself is partly based on \cite[Theorem 2]{kueng2017low}. Note that the same proof structure was also used in \cite{kueng2016low}.\\
	We start by introducing the Frobenius-robust rank null space property which was first introduced and used in \cite{kabanava2016stable}.
	\begin{definition}
	For $q\geq 1$, we say that $\mathcal{A}\colon \mathcal{H}_n\to \C^m$ satisfies the Frobenius-robust rank null space property with respect to $\ell_q$ of order $r$ with constants $0<\rho<1$ and $\tau>0$ if for all $Z\in\mathcal{H}_n$ the inequality
	\begin{align*}
		\Vert Z_r\Vert_F \leq \frac{\rho}{\sqrt{r}}\Vert Z_c\Vert_{\ast} +\tau\Vert \mathcal{A}(Z)\Vert_{\ell_q}
	\end{align*}
	holds.
	\end{definition}

	In order to prove Theorem \ref{th:general_recovery_from_exact_3_designs} we will show that the measurement process $\mathcal{A}$ considered in Theorem \ref{th:general_recovery_from_exact_3_designs} satisfies the Frobenius-robust rank null space property with respect to $\ell_q$ of order $r$.\\
	Therefore, define the set
	\[
	T_{\rho,r}=\left\{Z\in \mathcal{H}_n\, :\,\Vert Z\Vert_F=1,\, \Vert Z_r\Vert_F>\frac{\rho}{\sqrt{r}}\Vert Z_c\Vert_{\ast} \right\}.
	\]
	The key ingredient in establishing the Frobenius-robust rank null space property is Mendelson's small ball method \cite{koltchinskii2015bounding, mendelson2015learning} as used in \cite{tropp2015convex,kabanava2016stable}.
	
	\begin{theorem}[\cite{kabanava2016stable}]\label{th:mendelson_small_ball}
		Fix a set $E\in \R^d$. Let $\phi$ be a random vector on $\R^d$, and let $\phi_1,\ldots,\phi_m$ be independent copies of $\phi$. For $\theta>0$ let
		\[
		Q_{\theta}(E;\phi)=\inf_{u\in E} \Prob\big(\vert\langle \phi,u \rangle\vert \geq \theta\big).
		\]
		Let $\varepsilon_1,\ldots,\varepsilon_m$ be independent Rademacher variables, independent from everything else, and define
		\[
		W_m(E;\phi)=\E\sup_{u\in E}\langle h, u\rangle,
		\]
		where $h=\frac{1}{\sqrt{m}}\sum_{j=1}^m \varepsilon_j \phi_j$. Then, for any $\theta>0$, $t\geq 0$ and $q\geq 1$,
		\[
		\inf_{u\in E}\Bigg(\sum_{j=1}^m\vert\langle \phi_j,u\rangle\vert^q\Bigg)^{\frac{1}{q}}\geq m^{\frac{1}{q}-\frac{1}{2}}(\theta\sqrt{m}Q_{2\theta}(E;\phi)-2W_m(E;\phi)-\theta t)
		\]
		with probability at least $1-e^{-2t^2}$.
	\end{theorem}
	
	We identify $\mathcal{H}_n$ with $\R^{n^2}$. Then, we want to use Theorem \ref{th:mendelson_small_ball} for the set $E=T_{\rho,r}$. Therefore, we need to bound
	\begin{align*}
		Q_{2\theta}=\inf_{Z\in T_{\rho,r}}\Prob\big(\vert \tr(aa^{\ast}Z)\vert\geq2\theta\big) \quad\text{and}\quad W_m=\E\sup_{Z\in T_{\rho,r}}\tr(HZ)
	\end{align*}
	with $a$ be randomly drawn from $\{p_i,\widetilde{w}_i\}_{1\leq i\leq N}$ and $H=\frac{1}{\sqrt{m}}\sum_{j=1}^m \varepsilon_j a_j a_j^{\ast}$.
	For $W_m$ we will use known bounds, while the lower bound for $Q_{2\theta}$ requires the following lemma, in which we use a variant of the Paley-Zygmund inequality tailored to the setting of $3$-designs. 
	\begin{lemma}\label{lem:Q_bound}
		Let $a$ be randomly drawn from an exact super-normalized $3$-design $\{p_i,\widetilde{w}_i\}_{1\leq i\leq N}$. For every $0\leq \theta \leq 1$ it holds
		\[
		Q_{\theta}(T_{\rho,r};aa^{\ast})=\inf_{Z\in T_{\rho,r}}\Prob\left(\vert\tr(a a^{\ast}Z)\vert\geq \theta \right)\geq (1-\theta^2)^3\frac{1}{36(1+(1+\rho^{-1})^2)r}.
		\]
		\begin{proof}
			Let $Z\in T_{\rho,r}$. Our key contribution is the observation that a variant of Payley Zygmund stated in Lemma \ref{lem:payley_zygmund_variant} can be used instead of the classic Payley Zygmund inequality. This makes it possible to give a lower bound for $Q_{\theta}$. For $p=\frac{3}{2}$ the inequality implies that for every $0\leq \theta\leq 1$ it holds
			\begin{align}\label{eq:payley_zygmund_used}
				\Prob\Big(\vert \tr(aa^{\ast} Z)\vert^2 \geq \theta^2 \E\vert\tr(aa^{\ast}Z)\vert^2 \Big)\geq (1-\theta^2)^3 \frac{(\E \vert \tr(aa^{\ast} Z)\vert^2)^{3}}{\left(\E \vert \tr(aa^{\ast} Z)\vert^3\right)^{2}}.
			\end{align}
		Thus, we are interested in bounds for the second and third moment. Following the arguments of \cite[Proposition 12]{kueng2017low} we get
			\begin{align*}
				\E\vert\tr(aa^{\ast}Z)\vert^2&=\E\tr(aa^{\ast}Z)^2=\sum_{i=1}^N p_i \tr(\widetilde{w}_i\widetilde{w}_i^{\ast}Z)^2=n(n+1)\sum_{i=1}^N p_i \tr\left((w_i w_i^{\ast}Z)^{\otimes 2}\right) \\
				&=n(n+1)\sum_{i=1}^N p_i \tr\left((w_i w_i^{\ast})^{\otimes 2}Z^{\otimes 2}\right)=n(n+1) \tr\left(\sum_{i=1}^N p_i(w_i w_i^{\ast})^{\otimes 2}Z^{\otimes 2}\right)\\
				&=2 \tr\left(P_{\text{Sym}^2}Z^{\otimes 2}\right),
			\end{align*}
			where we used the design property in the last step. Now, \cite[Lemma 17]{kueng2017low} and $\Vert Z\Vert_F=1$ implies that
			\begin{align}\label{eq:second_mom_bound}
				\E\vert \tr(aa^{\ast}Z)\vert^2=\tr(Z)^2+\tr(Z^2)=\tr(Z)^2+1\geq \max\{1,\tr(Z)^2\}.
			\end{align}
			We can not immediately adapt the arguments for the third moment since $\vert \tr(aa^{\ast} Z)\vert^3\neq  \tr(aa^{\ast} Z )^3$ in general. But we can calculate an upper bound using
			\begin{align*}
				\vert \tr(aa^{\ast}Z)\vert^3=\tr(aa^{\ast}Z)^2\cdot\vert \tr(aa^{\ast}Z)\vert\leq \tr(aa^{\ast}Z)^2  \tr(aa^{\ast}\vert Z\vert).
			\end{align*}
			This gives analogously to the second moment the bound
			\begin{align*}
				\E\vert\tr(aa^{\ast}Z)\vert^3&\leq\E\tr(aa^{\ast}Z)^2  \tr(aa^{\ast}\vert Z\vert)=\sum_{i=1}^N p_i \tr(\widetilde{w}_i\widetilde{w}_i^{\ast}Z)^2 \tr(\widetilde{w}_i\widetilde{w}_i^{\ast}\vert Z\vert)\\
				&=\sqrt[4]{(n+1)n}^{\,6}\sum_{i=1}^N p_i \tr\left((w_i w_i^{\ast}Z)^{\otimes 2}\otimes w_i w_i^{\ast}\vert Z\vert\right) \\
				&=\sqrt[4]{(n+1)n}^{\,6}\sum_{i=1}^N p_i \tr\left((w_i w_i^{\ast})^{\otimes 3}\left(Z^{\otimes 2}\otimes \vert Z\vert\right)\right)\\
				&=\sqrt[4]{(n+1)n}^{\,6} \tr\left(\sum_{i=1}^N p_i(w_i w_i^{\ast})^{\otimes 3}\left(Z^{\otimes 2}\otimes \vert Z\vert\right)\right)\\
				&=\frac{3!\sqrt[4]{(n+1)n}^{\,6}}{(n+2)(n+1)n} \tr\left(P_{\text{Sym}^3}\left(Z^{\otimes 2}\otimes \vert Z\vert\right)\right)\\
				&=\frac{\sqrt{n(n+1)}}{n+2}\, 6\tr\left( P_{\text{Sym}^3}\left(Z^{\otimes 2}\otimes \vert Z\vert\right)\right).
			\end{align*}
			The trace was explicitly calculated in Lemma \ref{lem:trace_of_sym_map}. Thus,
			\begin{align*}
				\E\vert\tr(aa^{\ast}Z)\vert^3\leq \frac{\sqrt{n(n+1)}}{n+2} \left(\tr(Z)^2\tr(\vert Z\vert)+2\tr(Z)\tr(Z\vert Z\vert)+\tr(Z^2)\tr(\vert Z\vert)+2\tr(Z^2\vert Z\vert)\right).
			\end{align*}
			Since $\Vert \vert Z\vert \Vert_F=\Vert Z\Vert_F=1$, Hölder's inequality implies $\tr(Z\vert Z\vert)\leq\Vert Z\Vert_F \Vert \vert Z\vert \Vert_F=1$ and the monotonicity of the Schatten-norm gives $\tr(Z^2 \vert Z\vert)=\tr(\vert Z\vert^3)\leq \Vert Z\Vert_F^3=1$. Since $Z\in T_{\rho,r}$, \cite[Lemma 3.4]{kabanava2016stable} implies that we can write $Z=\sqrt{1+(1+\rho^{-1})^2}\sum_{j=1}^L t_j Y_j$ with $t_j\in[0,1]$, $\sum_{j=1}^L t_j=1$, $Y_j\in\mathcal{H}_n$, $\Vert Y_j\Vert_F=1$ and $\text{rank}(Y_j)\leq r$. Thus, 
			\begin{align}\label{eq:trace_inequality_for_T_rho_r}
				\tr(\vert Z\vert)&=\sqrt{1+(1+\rho^{-1})^2}\tr\left(\left\vert \sum_{j=1}^L t_j Y_j\right\vert\right)\leq \sqrt{1+(1+\rho^{-1})^2}\sum_{j=1}^L t_j \tr\left(\left\vert  Y_j\right\vert\right)\nonumber\\
				&\leq \sqrt{1+(1+\rho^{-1})^2}\sum_{j=1}^L t_j \,\sqrt{\text{rank}(Y_j)} \Vert Y_j\Vert_F\leq\sqrt{1+(1+\rho^{-1})^2}  \sqrt{r},
			\end{align}
			where we used that $\Vert\cdot\Vert_{\ast}=\tr(\vert\cdot\vert)$ defines a norm. Lastly, it holds $\frac{\sqrt{n(n+1)}}{n+2}\leq 1$. Putting everything together gives
			\begin{align}\label{eq:third_mom_bound}
				\E \vert\tr(aa^{\ast} \vert Z\vert)\vert^3 &\leq \sqrt{1+(1+\rho^{-1})^2}\sqrt{r} \tr(Z)^2+2\tr(Z)+\sqrt{1+(1+\rho^{-1})^2}\sqrt{r}+2\nonumber\\
				&\leq 6\sqrt{1+(1+\rho^{-1})^2}\sqrt{r} \max\{1,\tr(Z)^2\}.
			\end{align}
			The equations (\ref{eq:payley_zygmund_used}),(\ref{eq:second_mom_bound}) and (\ref{eq:third_mom_bound}) together imply
			\begin{align*}
				\Prob\left( \vert\tr(aa^{\ast}Z)\vert \geq \theta\right)&\geq \Prob\Big(\vert \tr(aa^{\ast} Z)\vert^2 \geq \theta^2 \E\vert\tr(aa^{\ast}Z)\vert^2 \Big)\\
				&\geq (1-\theta^2)^3 \frac{\left(\max\{1,\tr(Z)^2\}\right)^{3}}{\left(6\sqrt{1+(1+\rho^{-1})^2}\sqrt{r} \max\{1,\tr(Z)^2\}\right)^{2}}\\
				&\geq (1-\theta^2)^3\frac{1}{36(1+(1+\rho^{-1})^2)r}
			\end{align*}
			for every $0\leq\theta\leq 1$.
		\end{proof}
	\end{lemma}
	
	The following lemma is contained in the proof of \cite[Theorem 1.1]{kabanava2016stable}.
	
	\begin{lemma}\label{lem:W_m_bound}
		Let $a$ be randomly drawn from an exact super-normalized $3$-design $\{p_i,\widetilde{w}_i\}_{1\leq i\leq N}$. It holds
		\begin{align*}
			W_m(T_{\rho,r};aa^{\ast})=\E\sup_{Z\in T_{\rho,r}}\tr(HZ)\leq \sqrt{1+(1+\rho^{-1})^2} \sqrt{r} \E\Vert H\Vert_{\infty},
		\end{align*}
		where $H=\frac{1}{\sqrt{m}}\sum_{j=1}^m \varepsilon_j a_j a_j^{\ast}$.
		\begin{proof}
			We use Hölder's inequality and (\ref{eq:trace_inequality_for_T_rho_r}) to obtain
			\begin{align*}
				W_m=\E\sup_{Z\in T_{\rho,r}}\tr(HZ)\leq \sup_{Z\in T_{\rho,r}} \Vert Z\Vert_{\ast}\, \E\Vert H\Vert_{\infty}\leq \sqrt{1+(1+\rho^{-1})^2} \sqrt{r} \E\Vert H\Vert_{\infty}.
			\end{align*}
			
		\end{proof}
	\end{lemma}
	
	We are now well equipped to prove Theorem \ref{th:general_recovery_from_exact_3_designs}.
	
	\begin{proof}[Proof of Theorem \ref{th:general_recovery_from_exact_3_designs}]
	Mendelson' small ball method as described above together with Lemma \ref{lem:Q_bound} and Lemma \ref{lem:W_m_bound} gives for $0<\theta\leq \frac{1}{2}$, $t\geq 0$ and $q\geq 1$,
	\begin{align*}
		\inf_{Z\in T_{\rho,r}} \left(\sum_{j=1}^m \vert \tr(a_j a_j^{\ast}Z)\vert^q\right)^{\frac{1}{q}}&\geq m^{\frac{1}{q}-\frac{1}{2}}(\theta\sqrt{m}Q_{2\theta}-2W_m-\theta t)\\
		&\geq m^{\frac{1}{q}-\frac{1}{2}} \left(\frac{\theta \sqrt{m} (1-\theta^2)^3}{36(1+(1+\rho^{-1})^2)r}-2\sqrt{1+(1+\rho^{-1})^2} \sqrt{r} \E\Vert H\Vert_{\infty}-\theta t\right)
	\end{align*}	
	with probability at least $1-e^{-2t^2}$. 
	Choosing $\theta=\frac{1}{4}$ and $t=C_4\frac{\sqrt{m}}{\rho^{-2}r}$ with $C_4$ small enough as well as $C_1$ with $m\geq C_1 \rho^{-6} r^3 n\log(2n)$ large enough and using \cite[Proposition 13]{kueng2017low}, gives
	\begin{align*}
		&m^{\frac{1}{q}-\frac{1}{2}}\left(\frac{\theta \sqrt{m} (1-\theta^2)^3}{36(1+(1+\rho^{-1})^2)r}-2\sqrt{1+(1+\rho^{-1})^2} \sqrt{r} \E\Vert H\Vert_{\infty}-\theta t\right)\\
		&\geq\frac{3375 m^{\frac{1}{q}}}{589824(1+(1+\rho^{-1})^2)r}-2\cdot3.1049 \sqrt{1+(1+\rho^{-1})^2}\sqrt{r}\sqrt{n\log(2n)}m^{\frac{1}{q}-\frac{1}{2}}- \frac{C_4}{4\rho^{-2}r}m^{\frac{1}{q}}\\
		&\geq\frac{3375 m^{\frac{1}{q}}}{589824(1+(1+\rho^{-1})^2)r}-6.2098 \frac{\sqrt{1+(1+\rho^{-1})^2}}{\rho^{-3}}\frac{1}{\sqrt{C_1} r} m^{\frac{1}{q}}- \frac{C_4}{4\rho^{-2}r}m^{\frac{1}{q}}\\
		&\geq\frac{m^{\frac{1}{q}}}{\rho^{-2}r}\left(\frac{3375}{5\cdot 589824}-\frac{6.2098\cdot \sqrt{5}}{\sqrt{C_1}}-\frac{C_4}{4}\right)\geq \frac{m^{\frac{1}{q}}}{C_3\rho^{-2}r},
	\end{align*}
	where we used $1+(1+\rho^{-1})^2\leq 5\rho^{-2}$ for all $0<\rho<1$ in the last line. In summary, we have
	\begin{align*}
		\inf_{Z\in T_{\rho,r}} \left(\sum_{j=1}^m \vert \tr(a_j a_j^{\ast}Z)\vert^q\right)^{\frac{1}{q}}\geq  \frac{m^{\frac{1}{q}}}{C_3\rho^{-2}r}
	\end{align*}
	with probability at least $1-\exp\left(-\frac{2C_4^2m}{\rho^{-4}r^2}\right)$. Then, \cite[Lemma 3.3]{kabanava2016stable} says that $\mathcal{A}$ satisfies the Frobenius-robust rank null space property with respect to $\ell_q$ of order $r$ with constants $\rho$ and $\frac{C_3\rho^{-2}r}{m^{\frac{1}{q}}}$. The result now follows from \cite[Theorem 3.1]{kabanava2016stable}\footnote{Strictly speaking, we use a variant of \cite[Theorem 3.1]{kabanava2016stable} for Hermitian matrices, as described in \cite[Remark 3.2]{kabanava2016stable}.} which in a nutshell states that if $\mathcal{A}$ satisfies the Frobenius-robust rank null space property, then nuclear norm minimization is stable and robust.
	\end{proof}
	
	\subsection{Proof of Theorem \ref{th:general_recovery_from_approx_3_designs}}
	
	Again, we will prove Theorem \ref{th:general_recovery_from_approx_3_designs} in the setting described in Remark \ref{rem:equivalent_formulation_for_super_normalized_design}. Most of the proof parallels the exact design case. Inspecting the proof of Lemma \ref{lem:W_m_bound} shows that it is true for approximate designs as well. It remains to prove a generalized version of Lemma \ref{lem:Q_bound}. This can be done by adapting \cite[Generalized version of Proposition 12]{kueng2017low}.
	
	\begin{lemma}\label{lem:gen_bound_for_Q}
		Let the vector $a$ be randomly drawn from an approximate super-normalized $3$-design $\{p_i,\widetilde{w}_i\}_{1\leq i\leq N}$ of accuracy either $\theta_{1}\leq \frac{1}{4}$ or $\theta_{\infty}\leq \frac{1}{4r(1+(1+\rho^{-1})^2)}$. For every $0\leq \theta \leq 1$ it holds
		\[
		Q_{\theta}(T_{\rho,r};aa^{\ast})=\inf_{Z\in T_{\rho,r}}\Prob\left(\vert\tr(a a^{\ast}Z)\vert\geq \theta \right)\geq (1-\theta^2)^3\frac{1}{450
			(1+(1+\rho^{-1})^2)r}.
		\]
		\begin{proof}
			Let $Z\in T_{\rho,r}$. We start by considering the case where $\theta_{\infty}\leq \frac{1}{4r(1+(1+\rho^{-1})^2)}$. Paralleling the proof of \cite[Generalized version of Proposition 12]{kueng2017low} we get the following second moment bound
			\begin{align}\label{eq:second_mom_bound_approx_design}
				\E\vert\tr(aa^{\ast}Z)\vert^2&=(n+1)n\tr\left(\sum_{i=1}^Np_i (w_iw_i^{\ast})^{\otimes 2} Z^{\otimes 2}\right)\nonumber\\
				&=2\tr(P_{\sym^2} Z^{\otimes})+ (n+1)n\tr\left(\left(\sum_{i=1}^Np_i (w_iw_i^{\ast})^{\otimes 2}-\binom{n+1}{2}^{-1}P_{\sym^2}\right) Z^{\otimes 2}\right)\nonumber\\
				&\geq 2\tr(P_{\sym^2} Z^{\otimes})- (n+1)n \left\Vert \sum_{i=1}^Np_i (w_iw_i^{\ast})^{\otimes 2}-\binom{n+1}{2}^{-1}P_{\sym^2}\right\Vert_{\infty} \Vert Z^{\otimes 2}\Vert_{\ast}\nonumber\\
				&\geq 2\tr(P_{\sym^2} Z^{\otimes})-2\theta_{\infty} (1+(1+\rho^{-1})^2)r,
			\end{align}
			where we used Hölder's inequality, the approximate design property and (\ref{eq:trace_inequality_for_T_rho_r}). Now, \cite[Lemma 17]{kueng2017low} and $\theta_{\infty}\leq \frac{1}{4r(1+(1+\rho^{-1})^2)}$ gives
			\begin{align*}
				\E\vert\tr(aa^{\ast}Z)\vert^2\geq \tr(Z^2)+\tr(Z)^2-2\theta_{\infty} (1+(1+\rho^{-1})^2)r \geq \frac{1}{2}+\tr(Z)^2\geq \frac{1}{2}\max\{1,\tr(Z)^2\}.
			\end{align*}
			Instead of fourth moment bound (as in \cite{kueng2017low}) we will now derive a third moment bound and use the variant of Payley Zygmund. Likewise to (\ref{eq:second_mom_bound_approx_design}) we get
			\begin{align}\label{eq:third_mom_bound_approx_design}
				\E\vert \tr(aa^{\ast}Z)\vert^3 &\leq \sqrt[4]{(n+1)n}^{\,6} \tr\left(\sum_{i=1}^N p_i(w_i w_i^{\ast})^{\otimes 3}\left(Z^{\otimes 2}\otimes \vert Z\vert\right)\right)\nonumber\\
				&\leq \sqrt[4]{(n+1)n}^{\,6} \left\Vert \sum_{i=1}^N p_i(w_i w_i^{\ast})^{\otimes 3}-\binom{n+2}{3}^{-1}P_{\sym^3} \right\Vert_{\infty} \left\Vert Z^{\otimes 2}\otimes \vert Z\vert  \right\Vert_{\ast}\nonumber\\
				&\quad+6\frac{\sqrt{n(n+1)}}{n+2}\tr\left(P_{\sym^3}Z^{\otimes 2}\otimes \vert Z\vert\right)\nonumber\\
				&\leq 6\frac{\sqrt{n(n+1)}}{n+2} \left(\theta_{\infty}\sqrt{(1+(1+\rho^{-1})^2)r}^3 +\tr\left(P_{\sym^3}Z^{\otimes 2}\otimes \vert Z\vert\right)\right).
			\end{align}
			A similar calculation as in (\ref{eq:third_mom_bound}) shows that
			\begin{align*}
				\E\vert\tr(aa^{\ast}Z)\vert^3&\leq 6 \theta_{\infty}\sqrt{(1+(1+\rho^{-1})^2)r}^3+6\sqrt{1+(1+\rho^{-1})^2}\sqrt{r} \max\{1,\tr(Z)^2\}\\
				&\leq \frac{15}{2}\sqrt{(1+(1+\rho^{-1})^2)r}\max\{1,\tr(Z)^2\},
			\end{align*}
			where we also used $\theta_{\infty}\leq \frac{1}{4r(1+(1+\rho^{-1})^2)}$. The same arguments as in Lemma \ref{lem:Q_bound} conclude the proof.\\
			The case $\theta_{1}\leq \frac{1}{4}$ follows analogously by just using Hölder's inequality in (\ref{eq:second_mom_bound_approx_design}) and (\ref{eq:third_mom_bound_approx_design}) the other way around, i.e.
			\begin{align*}
				&(n+1)n\tr\left(\left(\sum_{i=1}^Np_i (w_iw_i^{\ast})^{\otimes 2}-\binom{n+1}{2}^{-1}P_{\sym^2}\right) Z^{\otimes 2}\right)\\
				&\geq - (n+1)n \left\Vert \sum_{i=1}^Np_i (w_iw_i^{\ast})^{\otimes 2}-\binom{n+1}{2}^{-1}P_{\sym^2}\right\Vert_{\ast} \Vert Z^{\otimes 2}\Vert_{\infty}\\
				&\geq -2\theta_{1} \Vert Z\Vert_F^2=-\frac{1}{2}
			\end{align*}
			and
			\begin{align*}
				&\sqrt[4]{(n+1)n}^{\,6}\tr\left(\left(\sum_{i=1}^Np_i (w_iw_i^{\ast})^{\otimes 3}-\binom{n+2}{3}^{-1}P_{\sym^3}\right) Z^{\otimes 2}\otimes \vert Z\vert \right)\\
				&\leq  \sqrt[4]{(n+1)n}^{\,6}\left\Vert \sum_{i=1}^Np_i (w_iw_i^{\ast})^{\otimes 3}-\binom{n+2}{3}^{-1}P_{\sym^3}\right\Vert_{\ast} \Vert Z^{\otimes 2}\otimes \vert Z\vert\Vert_{\infty}\\
				&\leq 6\frac{\sqrt{n(n+1)}}{n+2}\theta_{1}\Vert Z\Vert_F^3= \frac{3}{2}\frac{\sqrt{n(n+1)}}{n+2}.
			\end{align*}
		\end{proof}
	\end{lemma}

	\begin{proof}[Proof of Theorem \ref{th:general_recovery_from_approx_3_designs}]
		The proof follows word for word the proof of Theorem \ref{th:general_recovery_from_exact_3_designs} except that we use Lemma \ref{lem:gen_bound_for_Q} instead of Lemma \ref{lem:Q_bound} and \cite[Generalized version of Proposition 13]{kueng2017low} instead of \cite[Proposition 13]{kueng2017low}.
	\end{proof}
	
	\appendix
	\section{}
	
	\subsection{Payley-Zygmund inequality}
	
	We will present a proof for self-containment.
	
	\begin{lemma}\label{lem:payley_zygmund_variant}
		Let $p>1$, $0\leq \theta\leq 1$ and $W\geq 0$ a random variable with $\E W^p <\infty$. Then, it holds
		\begin{align*}
			\Prob\left(W> \theta \, \E W\right)\geq (1-\theta)^{\frac{p}{p-1}}\frac{(\E W)^{\frac{p}{p-1}}}{\left(\E W^p\right)^{\frac{1}{p-1}}}.
		\end{align*}
	\end{lemma}
	\begin{proof}
		We can write the expected value as
		\begin{align*}
			\E W=\E \left(W \mathbbm{1}_{\{W\leq \theta \E W\}}\right)+\E \left(W \mathbbm{1}_{\{W> \theta \E W\}}\right).
		\end{align*}
		The first summand is bounded by $\theta \E W$. For the second, we can use Hölder's inequality to obtain
		\begin{align*}
			\E \left(W \mathbbm{1}_{\{W> \theta \E W\}}\right)\leq \left(\E W^p\right)^{\frac{1}{p}} \left( \E {\,\mathbbm{1}_{\{W>\theta \E W\}}}^{\frac{p}{p-1}}\right)^{\frac{p-1}{p}}=\left(\E W^p\right)^{\frac{1}{p}} \Prob(W>\theta \E W)^{\frac{p-1}{p}}.
		\end{align*}
		Putting everything together yields
		\begin{align*}
			\E W\leq \theta \E W + \left(\E W^p\right)^{\frac{1}{p}} \Prob(W>\theta \E W)^{\frac{p-1}{p}}\quad\Leftrightarrow \quad \Prob\left(W> \theta \, \E W\right)\geq (1-\theta)^{\frac{p}{p-1}}\frac{(\E W)^{\frac{p}{p-1}}}{\left(\E W^p\right)^{\frac{1}{p-1}}}.
		\end{align*}
	\end{proof}
	
	\subsection{The trace of the symmetrizer map}
	
	Within our arguments we need the following calculation of the trace of the symmetrizer map $P_{\text{Sym}^3}$.
	
	\begin{lemma}\label{lem:trace_of_sym_map}
		Let $X,Y,Z\in\mathcal{H}_n$.  Then,
		\begin{align*}
			3!\tr\left(P_{\sym^3} X\otimes Y\otimes Z\right)=&\tr(X)\tr(Y)\tr(Z)+\tr(X)\tr(YZ)+\tr(Y)\tr(XZ)\\
			&+\tr(Z)\tr(XY)+\tr(XYZ)+\tr(XZY).
		\end{align*}
	\begin{proof}
		Let $X,Y,Z\in\mathcal{H}_n$. A basis of $(\C^n)^{\otimes 3}$ is given by $\{e_i\otimes e_j\otimes e_k\}_{1\leq i,j,k\leq n}$. For these basis elements we get
		\begin{align*}
			(X\otimes Y\otimes Z)(e_i\otimes e_j\otimes e_k)=Xe_i\otimes Ye_j\otimes Ze_k=\sum_{r,s,t=1}^n X_{ri}Y_{sj}Z_{tk}\cdot (e_r\otimes e_s\otimes e_t)
		\end{align*}
		and
		\begin{align*}
			&\left(P_{\text{Sym}^3} X\otimes Y\otimes Z\right)(e_i\otimes e_j\otimes e_k)\\
			&=\frac{1}{3!}(Xe_i\otimes Ye_j\otimes Ze_k+Xe_i\otimes Ze_k\otimes Ye_j+Ye_j\otimes Xe_i\otimes Ze_k+Ye_j\otimes Ze_k\otimes Xe_i\\
			&\,\,\,\,\,\,\,+Ze_k\otimes Xe_i\otimes Ye_j+Ze_k\otimes Ye_j\otimes Xe_i)\\
			&=\frac{1}{3!}\sum_{r,s,t=1}^n(X_{ri} Y_{sj} Z_{tk}+X_{ri} Z_{sk} Y_{tj}+Y_{rj}X_{si}Z_{tk}+Y_{rj}Z_{sk}X_{ti}+Z_{rk}X_{si}Y_{tj}+Z_{rk}Y_{sj}X_{ti})\\
			&\quad\quad\quad\quad\quad\cdot(e_r\otimes e_s\otimes e_t).
		\end{align*}
		The trace is invariant under basis transformations. This makes computing the trace easy since we know the range of each basis element under $P_{\sym^3} X\otimes Y\otimes Z$. Hence,
		\begin{align*}
			3!\tr\left(P_{\sym^3} X\otimes Y\otimes Z\right)&=\sum_{i,j,k=1}^n X_{ii}Y_{jj}Z_{kk}+X_{ii}Z_{jk}Y_{kj}+Y_{ij}X_{ji}Z_{kk} +Y_{ij}Z_{jk}X_{ki}+Z_{ik}X_{ji}Y_{kj}\\
			&\quad\quad\quad\quad+Z_{ik}Y_{ji}X_{ki}\\
			&=\tr(X)\tr(Y)\tr(Z)+\tr(X)\tr(YZ)+\tr(Y)\tr(XZ)+\tr(Z)\tr(XY)\\
			&\quad+\tr(XYZ)+\tr(XZY).
		\end{align*}
	\end{proof}
	\end{lemma}
	
	\section*{Acknowledgements}
	The author thanks H. Führ (RWTH University) for his helpful comments and advice. The author acknowledges funding by the Deutsche Forschungsgemeinschaft (DFG, German Research Foundation) - Project number 442047500 through the Collaborative Research Center "Sparsity and Singular Structures" (SFB 1481).
	\bibliography{References_3-designs.bib}
	\bibliographystyle{plain}
	\vspace{0.5cm}

\end{document}